%% file: main_Holant_J.tex
\documentclass[12pt, onecolumn]{IEEEtran} 
\usepackage{setspace}
\usepackage{graphicx,dblfloatfix}
\usepackage{graphicx}
\usepackage{color}
\usepackage{amssymb}
\usepackage{epsfig}
\usepackage{mathabx}

\newcommand{\std}{\mbox{std}}

\newcommand{\wch}{\widecheck{h}}

\newcommand{\whf}{\widehat{f}}
\newcommand{\whg}{\widehat{g}}

\newcommand{\C}{\mathbb{C}}

\newcommand{\A}{{\cal A}}

\newcommand{\be}{\beta}
\newcommand{\wtb}{\widetilde{\beta}}

\newcommand{\ote}{\! \otimes \!}
\newcommand{\ot}{\! \odot \!}
\newcommand{\Ot}{\bigodot}

\newcommand{\e}{\{e\}}

\newcommand{\E}{{\cal E}}
\newcommand{\B}{{\cal B}}

\newtheorem{Def}{\underline{\underline{Definition}}}
\newtheorem{Lemma}{{Lemma}}
\newtheorem{Theorem}{\underline{Theorem}}

\newtheorem{Prop}{{Proposition}}

\newtheorem{Cor}{Corollary}

\begin{document}
\bibliographystyle{IEEEtran} 

\title{On Holant Theorem and Its Proof}  
\author{Ali Al-Bashabsheh, Yongyi Mao and Abbas Yongacoglu 
}

\date{}
\maketitle

\input{abstract}


\input{intro}


\input{preliminaries}

\input{holant}

\input{conclusion}

\input{acknowledgment}

\bibliography{C:/Research_Latex/bibliographys/Holographic}
\end{document}

%% file: abstract.tex
\begin{abstract}

Holographic algorithms are a recent breakthrough in computer science and has found applications in information theory.  
This paper provides a proof to the central component of holographic algorithms, namely, the Holant theorem. Compared with previous works, the 
proof appears simpler and more direct. Along the proof, we also develop a mathematical tool, which we call c-tensor. We expect the notion of c-tensor may be applicable over a wide range of analysis.

\end{abstract} 

%% file: intro.tex
\section{Introduction}

Holographic algorithms \cite{Valiant2004:Holographic, Cai2007:Art, Cai2007:Symmetric, Cai2008:Bases-Collapse}, pioneered by L. Valiant,  are a recent breakthrough in the theory of computation. Specifically, in his landmark paper \cite{Valiant2004:Holographic}, Valiant presented polynomial-time
 algorithms -- holographic algorithms --- for a number of problem families which were not known to be in P previously. Holographic algorithms have been recently introduced to the information theory community and applied to solving the capacity of certain constrained coding problems \cite{Schwartz:2008}.

Briefly, holographic algorithms set out to compute the sum of a multi-variate function over its configuration space where the function considered involves a large number of variables and factors as the product of local functions. Problems of such nature arise frequently in information theory, for example in computing constrained-coding capacity and in decoding various error correction codes. It is well-known that in general, obtaining exact solutions for such problems is computationally intractable. However, in the methodology of holographic algorithms, when it is possible to apply certain transformation of the problems, polynomial-time solvers can be constructed. Such transformations, which Valiant refers to as holographic reductions, form the basis of holographic algorithms.

The central component of holographic reductions is the {\em Holant Theorem}, which was introduced and proved in \cite{Valiant2004:Holographic}. The proof however appears rather encrypted to many audience,  and subsequently inspired an alternative proof given by Cai and  Choudhary in \cite{Cai2007:Tensor}. As the proof of \cite{Cai2007:Tensor} uses a sophisticated machinery of tensors that is not familiar to broad audience, here we present a more direct proof. At least to us, our proof seems simpler, accessible for general audience and contains more insights. Additionally, the Holant Theorem in this paper is proved in its most general form, namely, the variables may take values from arbitrary alphabet. This contrasts the previous works where only binary alphabets are considered.

It  is remarkable that our proof relies on the notion of ``c-tensor'', a term which we coin in this paper. In a sense equivalent to the standard tensor product, c-tensor differs from the standard tensor product in that it is a commutative operation. Using c-tensor, our proof of Holant Theorem appears more transparent. Although perhaps under the guise of mathematical literature, the notion of c-tensor synthesized in this paper appears to be a useful tool. We expect that c-tensor may find other applications in context beyond Holant Theorem.

This paper is organized as follows. Section \ref{sec:pre} gives mathematical preliminaries, where the main focus is to develop the mathematical tool of c-tensor. Section \ref{sec:holant} states and proves Holant theorem. Section \ref{conclusion} provides a brief conclusion.

%% file: preliminaries.tex
\section{Preliminaries}

\label{sec:pre}

\subsection{Assignments}
For an arbitrary finite set%
\footnote{All sets in this work are non-empty unless otherwise specified.}
$E$ and any finite alphabet ${\cal A}$, we refer to any function mapping $E$ into ${\cal A}$ as an
$\cal A$-\emph{assignment} on $E$, and collectively denote the set of all such assignments by ${\cal A}^{E}$.
If $x$ is an $\A$-assignment on $E$, we often write  $x_E$ in place of $x$ to make explicit the domain of function $x$. 
For any subset $U \subseteq E$, the \emph{restriction}
of an ${\cal A}$-assignment $x_{E}$ to $U$ will be denoted by $x_{E:U}$. That is, $x_{E:U}$ is an ${\cal A}$-assignment on $U$ such
that for every $e\in U$, $x_{E:U}(e)=x_E(e)$.
%
%
%
If $E$ is clear from the context, we may write $x_{U}$ in place of $x_{E:U}$ for simplicity. In particular, such practice will be
more common when $U$ is a singleton $\{e\}$ for some $e\in E$. In this case, we always write $x_{\{e\}}$ rather than $x_{E:\{e\}}$.

With a slight abuse of notation, an assignment, say $x_E$, is also treated as a {\em set}, namely, the set $\{(e, x_E(e))|e\in E\}$
or the graph of function $x_E$.
Under such interpretation, restriction $x_{E:U}$ is also understood as set $\{(e,x_{E}(e)) | e \in U\}$. In addition, for any two disjoint finite sets
$E_1$ and $E_2$ and any two ${\cal A}$-assignments $x_{E_1}$ and $x_{E_2}$, the union $x_{E_1}\cup x_{E_2}$ is well defined and can be interpreted back as an ${\cal A}$-assignment on $E_1\cup E_2$ defined by
\[
(x_{E_1}\cup x_{E_2})(e) = \left\{ \begin{array}{cc}
x_{E_{1}}(e), & \mbox{if } e \in E_{1} \\
x_{E_{2}}(e), & \mbox{if }  e \in E_{2}
\end{array} \right.
\]
Conversely, every ${\cal A}$-assignment $x_{E_1 \cup E_2}$ on $E_1\cup E_2$ can be understood as the union $x_{E_1}\cup x_{E_2}$ of two ${\cal A}$-assignments $x_{E_1}$ and $x_{E_2}$, which are the restrictions of $x_{E_1\cup E_2}$ to $E_1$ and $E_2$, respectively.
Furthermore, it is easy to verify that the decomposition of any $x_{E_1\cup E_2}$ in terms of the union of two such restrictions is unique for any fixed choice of $E_1$ and $E_2$. This establishes a one-to-one correspondence between the set ${\cal A}^{E_1\cup E_2}$
of all ${\cal A}$-assignments on $E_1\cup E_2$ and the cartesian product ${\cal A}^{E_1}\times {\cal A}^{E_2}$. Extending this argument by induction (on the number of disjoint sets), we obtain the following lemma. 

\begin{Lemma}
Let $E$ be a finite set and $\{ E_{1}, \ldots, E_{p}  \}$  be an arbitrary partition of $E$. Then for any finite alphabet ${\cal A}$ and
any ${\cal A}$-assignment $x_E$ on $E$, there exists a unique sequence of restrictions $x_{E:E_1}, x_{E:E_2}, \ldots, x_{E:E_p}$
of $x_E$ to $E_1, E_2, \ldots E_p$ respectively
such that
$
x_E=\bigcup_{i=1}^{p}x_{E:E_i}.
$
Conversely, for any given $\cal A$-assignments $x_{E_1}, x_{E_2}, \ldots x_{E_p}$, the union $\bigcup_{i=1}^{p}x_{E_i}$ is an
$\cal A$-assignment on $E$. That is, there is a one-to-one correspondence between ${\cal A}^E$ and the $p$-fold cartesian
product ${\cal A}^{E_1}\times {\cal A}^{E_2}\times \ldots \times{\cal A}^{E_p}$.
\label{Lemma:0}
\end{Lemma}



%

\subsection{Space of functions}

%
%

In this paper, we will often work  with vector spaces in the form of ${\C^{S}}$, where $\C^{S}$ is the set of all functions from a finite set $S$ to the field of complex numbers $\C$.
 The following lemma, straight-forward to prove, justifies that the set $\C^S$ is a vector space.
\begin{Lemma}
For any finite set $S$, let the set $\C^{S}$ of all functions  mapping $S$ into $\C$ be equipped with the following  two operations:
\begin{itemize}
\item Addition: $\forall f, f' \in \C^{S}$, $(f+f')(s): = f(s) + f'(s)$ for every $s\in S$.
\item Scalar multiplication: $\forall f\in \C^{S}$ and $\alpha\in \C$, $(\alpha f)(s):= \alpha(f(s))$ for every $s\in S$.
\end{itemize}
Then $\C^{S}$
 is a vector space isomorphic to
$\C^{|S|}$.
\label{Lemma:iso}
\end{Lemma}


Note that from this lemma it is immediate that $\C^{S} \cong \C^{S'}$ whenever $|S| = |S'| < \infty$. In the lemma, the fact that $\C^{S} \cong \C^{|S|}$ can be established by advising an explicit vector space isomorphism. 
Such isomorphism is clearly not unique. 
Let $r$ be an arbitrary set bijection from $S$ into the set $\{1, \ldots, |S| \}$. 
It is not hard to show that $\sigma: \C^{S} \rightarrow \C^{|S|}$ such that 
\[
\sigma(f) = ( f(r^{-1}(1)), \ f(r^{-1}(2)), \ \ldots, \ f(r^{-1}(|S|))   )   
\]
for all $f \in \C^{S}$ is a vector space isomorphism. We refer to such isomorphism as a \emph{natural} one.
At places where confusion is unlikely, we may write $f r^{-1} (\cdot)$ for the composition $f(r^{-1}(\cdot))$ to enhance readability.

Since $\C^{S}$ is a vector space over $\C$ of dimension $|S|$, one can find a set of functions which form a basis in $\C^{S}$. An obvious example of a basis is the ``standard" one, namely, the set
$\B_{\std} = \{\delta_{s} : s \in S\}$ where for any $s \in S$, $\delta_{s}: S \rightarrow \C$ is such that $\delta_{s}(x) = 1$ if $x = s$ and otherwise
$\delta_{s}(x) = 0$ for all $x \in S$.
Note that a natural vector space isomorphism mapping $\C^{S}$ to $\C^{|S|}$ simply accounts to taking coordinates with respect to $\B_{\std}$ under the ordering induced by the bijection $r$ associated with the natural isomorphism. More explicitly, let $\sigma: \C^{S} \rightarrow \C^{|S|}$ be a natural isomorphism and let $r:S \rightarrow \{1, \ldots, |S|\}$ be the bijection associated with $\sigma$. Then for any $f \in \C^{S}$, the $i$th component of $\sigma(f)$ is the coefficient of $\delta_{r^{-1}(i)}$ when $f$ is expressed as a linear combination of elements from the basis $\B_{\std}$.
Another example of basis can be constructed as follows. Let $S = \{s_{1}, \ldots, s_{|S|}\}$ and let
$S_{i} = \{s_{1}, \ldots, s_{i}\}$ for all $1 \leq i \leq |S|$. Further let $\B = \{\beta_{1}, \ldots, \beta_{|S|}\}$
where $\beta_{i}: S \rightarrow \C$ is such that $\beta_{i}(x) = 1$ if $x \in S_{i}$ and otherwise $\beta_{i}(x) = 0$ for all $x \in S$. Then, one can easily verify that the set $\B$ is a basis of $\C^{S}$.

Finally, for any vector space $\C^{S}$, as above,  we define the map $\langle \cdot, \cdot \rangle : \C^{S} \times \C^{S} \rightarrow \C$ by
\begin{eqnarray}
\langle f,f' \rangle: = \sum_{x \in S} f(x) f'(x)
\label{eq:inner}
\end{eqnarray}
for all $f, f' \in \C^{S}$.
It is clear that the map $\langle \cdot,\cdot \rangle$ is bilinear.

As we will see momentarily, vector space $\C^{S}$ considered in this paper mostly takes $S$
as ${\cal A}^{E}$ for some choice of alphabet $\cal A$ and finite set $E$. In this case, we will write $\C_{\cal A}^{E}$ in place of $\C^{({\cal A}^{E}  )}$ to lighten the notations.

\subsection{c-Tensors}

Let $E_{1}$ and $E_{2}$ be two disjoint sets. For any $f_{1} \in \C_{\cal A}^{E_{1}}$ and $f_{2} \in \C_{\cal A}^{E_{2}}$,
we define the {\em c-tensor} $f_1\ot f_2$ of $f_1$ and $f_2$ as the element in $\C^{E_1\cup E_2}_{\cal A}$ (that is, the function mapping ${\cal A}^{E_1\cup E_2}$ into $\C$) such that for every $x_{E_1\cup E_2}\in {\cal A}^{E_1\cup E_2}$,
\[
(f_1\ot f_2)(x_{E_1\cup E_2}):= f_1(x_{(E_1\cup E_2):E_1})f_2(x_{(E_1\cup E_2):E_2})
\]

Inductively, we can extend the 
 notion of c-tensor to arbitrary number of functions $f_1 \in {\C}_{\cal A}^{E_1}, f_2 \in {\C}_{\cal A}^{E_2}, \ldots, f_p \in {\C}_{\cal A}^{E_p}$, where $\{E_1, E_2, \ldots, E_p\}$ is an arbitrary collection of disjoint finite sets.
The following lemma is directly provable from the definition of c-tensor.

\begin{Lemma}
c-tensor is commutative and associative.
\end{Lemma}
\begin{proof}
Let $f_{1} \in \C_{\A}^{E_{1}}$, $f_{2} \in \C_{\A}^{E_{2}}$ and $f_{3} \in \C_{\A}^{E_{3}}$ where $E_{1}, E_{2}$ and $E_{3}$ are finite disjoint sets.
Then,
\begin{eqnarray*}
(f_{1} \ot f_{2}) \ot f_{3} (x_{E_{1} \cup E_{2} \cup E_{3} } ) \hspace{-.5cm} &&  = f_{1} \ot f_{2} (x_{E_{1} \cup E_{2} \cup E_{3} : E_{1} \cup E_{2} }) f_{3} (x_{E_{1} \cup E_{2} \cup E_{3} : E_{3} }) \\
&& = f_{1} (x_{E_{1} \cup E_{2} \cup E_{3} : E_{1} }) f_{2} (x_{E_{1} \cup E_{2} \cup E_{3} : E_{2} }) f_{3} (x_{E_{1} \cup E_{2} \cup E_{3} : E_{3} })\\
&& = f_{1} (x_{E_{1} \cup E_{2} \cup E_{3} : E_{1} }) f_{2} \ot f_{3} (x_{E_{1} \cup E_{2} \cup E_{3} : E_{2} \cup E_{3} })  \\
&& = f_{1} \ot ( f_{2} \ot f_{3} ) (x_{E_{1} \cup E_{2} \cup E_{3} } )
\end{eqnarray*}
Further,
\begin{eqnarray*}
f_{1} \ot f_{2} (x_{E_{1} \cup E_{2}}) \hspace{-.5cm}&& = f_{1}(x_{E_{1} \cup E_{2} : E_{1}}) f_{2}(x_{E_{1} \cup E_{2} : E_{2}}) \\
&&= f_{2}(x_{E_{1} \cup E_{2} : E_{2}}) f_{1}(x_{E_{1} \cup E_{2} : E_{1}}) \\
&&= f_{2} \ot f_{1} (x_{E_{1} \cup E_{2}})
\end{eqnarray*}
As claimed.
\end{proof}

As will be shown momentarily, c-tensor of two functions is in a sense equivalent to the standard notion of tensor product of two vectors. 
The reason we refer to this operation ``c-tensor'' is to emphasize its commutative nature, which in general does not hold for the standard tensor product. 
 Since the c-tensor is independent of bracketing and ordering (due to the previous lemma), the  notation $\Ot_{i = 1}^{p} f_{i}$, denoting the $p-$fold c-tensor of functions
 $f_1, f_2, \ldots, f_p$, is well defined.

We now show the equivalence between c-tensor and tensor product, where standard tensor product of two vectors $u$ and $v$ is denoted by  $u\otimes v$.

\begin{Lemma}
Let $E_1$ and $E_2$ be disjoint and $\sigma_1:{\C}_{\cal A}^{E_1}\rightarrow {\C}^{|{\cal A}^{E_1}|}$ and $\sigma_2:{\C}_{\cal A}^{E_2}\rightarrow {\C}^{|{\cal A}^{E_2}|}$ be two natural vector-space isomorphisms. Then there exists a natural vector-space isomorphism
$\sigma:{\C}_{\cal A}^{E_1\cup E_2}\rightarrow {\C}^{|{\cal A}^{E_1\cup E_2}|}$ such that
$\sigma(f_1\ot f_2)=\sigma_1(f_1)\otimes \sigma_2(f_2)$
for any $f_1\in \C_{\cal A}^{E_1}$ and
$f_2\in \C_{\cal A}^{E_2}$.
\end{Lemma}

\begin{proof}
Let $r_{1}: \A^{E_{1}} \rightarrow \{1, \ldots, |\A^{E_{1}}|\}$ and 
$r_{2}: \A^{E_{2}} \rightarrow \{1, \ldots, |\A^{E_{2}}|\}$ be the two bijections associated with 
$\sigma_{1}$ and $\sigma_{2}$, respectively. Define the set map 
$r: \A^{E_{1} \cup E_{2}} \rightarrow \{1, \ldots, |\A^{E_{1}}|\} \times \{1, \ldots, |\A^{E_{2}}|\}$ given by
$r(x_{E_{1} \cup E_{2}}) = (r_{1}( x_{(E_{1} \cup E_{2}):E_{1}} ), r_{2}(x_{(E_{1} \cup E_{2}):E_{2}})  )$. Then $r$ is a bijection
since $r_{1}$ and $r_{2}$ are bijections.
Now define $\sigma: \C_{\A}^{E_{1} \cup E_{2}} \rightarrow \C^{|\A^{E_{1}\cup E_{2}}|}$ as
\[
\sigma(f) = (fr^{-1}(1,1), \ fr^{-1}(1,2), \ \ldots, \ 
fr^{-1}(|\A^{E_{1}}|,|\A^{E_{2}}|)
)  
\]
for all $f \in \C_{\A}^{E_{1} \cup E_{2}}$. Then $\sigma$ is a natural isomorphism and
\begin{eqnarray*}
\sigma(f_{1} \ot f_{2})
&& \hspace{-.6cm}= \left( (f_{1} \ot f_{2}) (r^{-1}(1,1)), \ (f_{1} \ot f_{2})(r^{-1}(1,2)), \ldots, 
 (f_{1} \ot f_{2})(r^{-1}(|\A^{E_{1}}|,|\A^{E_{2}}|)) \right) \\
&&\hspace{-.6cm}= \left( f_{1}r^{-1}_{1}(1)  f_{2}r^{-1}_{2}(1), \ f_{1}r^{-1}_{1}(1)  f_{2}r^{-1}_{2}(2), 
\ldots,  
 f_{1}r^{-1}_{1}(1) f_{2}r^{-1}_{2}(|\A^{E_{2}}|), \right. \\
&&\hspace{.5cm} \hspace{-.6cm}  f_{1}r^{-1}_{1}(2)  f_{2}r^{-1}_{2}(1), \ f_{1}r^{-1}_{1}(2)  f_{2}r^{-1}_{2}(2), 
\ldots,  
  f_{1}r^{-1}_{1}(2) f_{2}r^{-1}_{2}(|\A^{E_{2}}|),  \\
&&\hspace{5cm} \ldots \\
&&\hspace{.5cm} \hspace{-.6cm}  f_{1}r^{-1}_{1}(|\A^{E_{1}}|)  f_{2}r^{-1}_{2}(1), \ f_{1}r^{-1}_{1}(|\A^{E_{1}}|)  f_{2}r^{-1}_{2}(2), 
\ldots,  
 \left. f_{1}r^{-1}_{1}(|\A^{E_{1}}|) f_{2}r^{-1}_{2}(|\A^{E_{2}}|) \right) \\
&&\hspace{-.6cm}= ( f_{1}r^{-1}_{1}(1) \sigma_{2}(f_{2}), \ldots, f_{1}r^{-1}_{1}(|\A^{E_{1}}|) \sigma_{2}(f_{2}) ) \\ 
&&\hspace{-.6cm}= \sigma_{1}(f_{1}) \otimes \sigma_{2}(f_{2})
\end{eqnarray*}
as desired.
\end{proof}

Extending this lemma by induction to multi-fold c-tensor gives the following corollary.

\begin{Cor}
Let $E_1, E_2, \ldots, E_p$ be pairwise disjoint and $\sigma_1:{\C}_{\cal A}^{E_1}\rightarrow {\C}^{|{\cal A}^{E_1}|},  \sigma_2:{\C}_{\cal A}^{E_2}\rightarrow {\C}^{|{\cal A}^{E_2}|}, \ldots, \sigma_p:{\C}_{\cal A}^{E_p}\rightarrow {\C}^{|{\cal A}^{E_p}|}$ be natural vector-space isomorphisms. Then there exists a natural vector-space isomorphism
$\sigma:{\C}_{\cal A}^{\bigcup_i^p E_i}\rightarrow {\C}^{|{\cal A}^{\bigcup_i^p E_i}|}$ such that
\[\sigma(\Ot_i^p f_i)=\sigma_1(f_1)\otimes \sigma_2(f_2) \otimes \ldots \otimes \sigma_p(f_p)\]
 for any $f_1\in \C_{\cal A}^{E_1}, f_2\in \C_{\cal A}^{E_2},\ldots, f_p\in \C_{\cal A}^{E_p}$.
\label{cor:iso}
\end{Cor}

This establishes a sense of equivalence between c-tensor and the standard tensor product.

\subsection{Basis}

%

Let $\A$ be a finite alphabet and $E$ be a finite set. For each $e \in E$,  define map $\tau_{e}:\C^{\A} \rightarrow \C_{\A}^{\e}$ such that
for all $f \in \C^{\A}$
\[
\tau_{e}(f)(x_{\e}) = f(x_{\e}(e))
\]
for all $x_{\e} \in \A^{\e}$.

The following lemma shows that $\tau_{e}$ is an isomorphism.
\begin{Lemma}
Let $E$ be a finite set and for each $e \in E$, define $\tau_{e}:\C^{\A} \rightarrow \C_{\A}^{\e}$ as above.
Then each $\tau_{e}$ is a vector space isomorphism. 
\label{lemma:sigma_e}
\end{Lemma}
\begin{proof}
Let $f$ and $f'$ be arbitrary functions from $\C^{\A}$ and $\alpha$ and $\alpha'$ be arbitrary scalars from $\C$. Then
for all $x_{\e} \in \A^{\e}$ we have
\begin{eqnarray*}
\tau_{e} (\alpha f + \alpha' f') (x_{\e}) \hspace{-.6cm} &&  = (\alpha f + \alpha' f') (x_{\e}(e)) \\
&& = (\alpha f) (x_{\e}(e)) + (\alpha' f') (x_{\e}(e)) \\
&& = \alpha f(x_{\e}(e)) + \alpha' f'(x_{\e}(e)) \\
&& = \alpha \tau_{e}(f)(x_{\e}) + \alpha' \tau_{e}(f) (x_{\e}) 
\end{eqnarray*}
Hence, $\tau_{e} (\alpha f + \alpha' f') = \alpha \tau_{e}(f) + \alpha' \tau_{e}(f')$ and therefore $\tau_{e}$ preserves
addition and scalar multiplication.

Assume $f,f' \in \C^{\A}$ are such that $\tau_{e}(f) = \tau_{e}(f')$ then $\tau_{e}(f) (x_{\e}) = \tau_{e}(f')(x_{\e})$
for all $x_{\e} \in \A^{\e}$. Hence, $f(x_{\e}(e)) = f'(x_{\e}(e))$ for all $x_{\e}$. This is equivalent to
$f(a) = f'(a)$ for all $a \in \A$ and hence $f = f'$, which implies $\tau_{e}$ is injective. 
Surjectivity of $\tau_{e}$ is automatic since $\C^{\A}$ and $\C_{\A}^{\e}$ both have the same finite dimension.
\end{proof}

By lemma \ref{Lemma:iso} we know that $\C^{\A}$ is a vector space of dimension $|\A|$. Let $\B$ be a basis%
\footnote{Our definition for a basis coincide with the one from linear algebra, i.e, a basis is a linearly independent spanning set.
In contrast, the definition of a basis in Valiant's work was that of a spanning set. This work still holds if the definition
of a basis was chosen to confine with Valiant's definition.}
 of $\C^{\A}$ and for each $e \in E$ denote by $\tau_{e}(\B)$ the image of $\B$ under $\tau_{e}$,
that is, $\tau_{e}(\B) = \{ \tau_{e}(\beta) | \beta \in \B \}$.
Then $\tau_{e}(\B)$ forms a basis in $\C_{\A}^{\e}$ for all $e \in E$ because an isomorphism maps basis into basis.
Since each element in the basis ${\tau_{e}(\B)}$ is a map from $\A^{\e}$ into $\C$,
then for any $\B$-assignment on $E$, say $b_{E}$,
the c-tensor $\Ot_{e \in E} \tau_{e}(b_{E}(e))$ is a map from $\A^{E}$ into $\C$. 
This motivates the following definition. For any $E$, $\A$ and $\B$ as above, we define the map $\E: \A^{E} \times \B^{E} \rightarrow \C$ such that
\[
\E(a_{E}, b_{E}) = ( \Ot_{e \in E} \tau_{e}(b_{E}(e) ) ) (a_{E})
\]
for all $(a_{E}, b_{E}) \in \A^{E} \times \B^{E}$.
Note that the name of the map was chosen to emphasize that its arguments are assignments on $E$. Hence in subsequent discussions-
if we write for instance $\E_{i}$, then we mean the map from $\A^{E_{i}} \times {\B}^{E_{i}}$ to $\C$ defined as above.

Now if $b_{E}$ is a fixed $\B$-assignment on $E$, then map $\E$ induces, via $b_{E}$, a function $\E^{[b_{E}]}$ mapping $\A^{E}$ to $\C$:
Formally, for any $b_{E} \in \B^{E}$ we define $\E^{[b_{E}]}: \A^{E} \rightarrow \C$ by
$
          \E^{[b_{E}]}( a_{E} ) = \E(a_{E}, b_{E})
$
for all $a_{E} \in \A^{E}$.

Conversely, if $a_{E}$ is a fixed $\A$-assignment on $E$, then map $\E$ induces, via $a_{E}$, a function $\E_{[a_{E}]}$ mapping $\B^{E}$ to $\C$:
Formally, for any $a_{E} \in \A^{E}$ we define $\E_{[a_{E}]} : \B^{E} \rightarrow \C$ by
$
          \E_{[a_{E}]}( b_{E} ) = \E(a_{E}, b_{E})
$
for all $b_{E} \in \B^{E}$.

%

Before we proceed, we need the following result from algebra 
\begin{Theorem}
If $\B = \{\be_{1}, \ldots, \be_{k} \}$ is a basis of $\C^{k}$,
then $\{\be_{i_{1}} \ote \ldots \ote \be_{i_{m}} | \be_{i_{1}}, \ldots, \be_{i_{m}} \in \B, 1 \leq i_{j} \leq k \ \forall j \}$ 
is a basis of the vector space $\C^{k^{m}}$.
\label{Thm:algebra}
\end{Theorem}

The following theorem is a recast of the previous one in the costumes of c-tensors and our frame of work.
\begin{Theorem}
If $\B = \{\be_{1}, \ldots, \be_{|\A|}\}$ is a basis of $\C^{\A}$, then $\{\E^{[b_{E}]} | b_{E} \in \B^{E}  \}$ is a basis of $\C_{\A}^{E}$.
\label{Thm:basis}
\end{Theorem}

\begin{proof}
First note that $\{\{e_{1}\}, \ldots, \{e_{|E|}\} \}$ forms a partition of $E$ of size $|E|$.
For each $e \in E$, let $\tau_{e} : \C^{\A} \rightarrow \C_{\A}^{\{e\}}$ be an isomorphism as in lemma \ref{lemma:sigma_e}.
Further, let $\sigma : \C^{\A} \rightarrow \C^{|A|}$ be a natural vector space isomorphism and define $\sigma_{e}: \C_{\A}^{\{e\}} \rightarrow \C^{|\A|}$ as the composition $\sigma_{e} = \sigma \circ \tau_{e}^{-1}$ for all $e \in E$. Then, $\sigma_{e}$ is a natural vector space isomorphism for all $e \in E$. 
Hence, by corollary \ref{cor:iso}, 
there exists  a (natural) isomorphism
$\sigma_{E}: \C_{\A}^{E} \rightarrow \C^{|\A^{E}|}$ such that for all $f_{1} \in \C_{\A}^{\{e_{1}\}}, \ldots, f_{|E|} \in \C_{\A}^{\{e_{|E|}\}}$,
\[
\sigma_{E}(\Ot_{i=1}^{|E|}f_{i}) = \sigma_{e_1}(f_{1}) \otimes \ldots \otimes \sigma_{e_{|E|}}(f_{|E|})
\]
Let $\sigma(\B) = \{ \wtb_{1}, \ldots, \wtb_{|\A|} \}$ be the image of $\B$ under $\sigma$, then $\sigma(\B)$ is a basis of $\C^{|\A|}$. 
Let $X = \{\E^{[b_{E}]} | b_{E} \in \B^{E}  \}$ and
$Y = \{\wtb_{i_{1}} \ote \ldots \otimes \wtb_{i_{|E|}} | \wtb_{i_{1}}, \ldots, \wtb_{i_{|E|}} \in \sigma(\B), 1 \leq i_{j} \leq |\A| \ \forall j \}$.
Then theorem \ref{Thm:algebra} asserts that $Y$ is a basis of $\C^{|\A^{E}|}$. Hence we are done if we can show that
$\sigma_{E}(X) = Y$, since this will imply that $X$ is a basis (due to the fact that an isomorphism maps basis to basis). To this end, note that
\begin{eqnarray*}
x \in X \hspace{-.5cm} && \Rightarrow x = \Ot_{e \in E} \tau_{e}(b_{E}(e) ), \mbox{ some } b_{E} \in \B^{E} \\
        &&\hspace{-0cm} \Rightarrow x = \tau_{e_{1}}(\be_{i_{1}}) \ot \ldots \ot \tau_{e_{|E|}}(\be_{i_{|E|}}), \mbox{ some } \be_{i_{1}}, \ldots, \be_{i_{|E|}} \in \B \\
        &&\hspace{-0cm} \Rightarrow \sigma_{E}(x) = \sigma_{e_1}(\tau_{e_{1}}(\be_{i_{1}})) \otimes \ldots \otimes \sigma_{e_{|E|}}(\tau_{e_{|E|}}(\be_{i_{|E|}})), \mbox{ some } \be_{i_{1}}, \ldots, \be_{i_{|E|}} \in \B \\ 
        &&\hspace{-0cm} \Rightarrow \sigma_{E}(x) = \sigma(\be_{i_{1}}) \otimes \ldots \otimes \sigma(\be_{i_{|E|}}), \mbox{ some } \be_{i_{1}}, \ldots, \be_{i_{|E|}} \in \B \\
 				&&\hspace{-0cm} \Rightarrow \sigma_{E}(x) = \wtb_{j_{1}} \otimes \ldots \otimes \wtb_{j_{|E|}}, \mbox{ some } \wtb_{j_{1}}, \ldots, \wtb_{j_{|E|}} \in \sigma(\B) 				 \\
 				&&\hspace{-0cm} \Rightarrow \sigma_{E}(x) \in Y
\end{eqnarray*}
Thus, $\sigma_{E}(X) \subseteq Y$. Since both $\sigma_{E}$ and $\sigma$ are bijective, it follows that
$|\sigma_{E}(X)| = |X| = |\B|^{|E|} = |\sigma(\B)|^{|E|} = |Y|$. Therefore, $\sigma_{E}(X) = Y$.
\end{proof}

The following definition is fundamental
\begin{Def}
Let $\B$ be a basis for $\C^{\A}$ and for any $f \in \C_{\A}^{E}$ define 
\begin{itemize}
\item $\widehat{f}$ as the unique%
\footnote{If $\B$ was a spanning set, then $\whf$ is defined as \emph{any} function in $\C_{\A}^{E}$ satisfying 
$
f(a) = \langle \whf, \E_{[a]} \rangle
$ for all $a \in \A^{E}$.}
 map from $\B^{E}$ into $\C$ satisfying for all $a \in \A^{E}$
\[
f(a) = \langle \whf, \E_{[a]} \rangle
\]

\item $\widecheck{f}$ as the map from $\B^{E}$ into $\C$ such that for all $b \in \B^{E}$,
\[\widecheck{f}(b) = \langle f,\E^{[b]} \rangle \]
\end{itemize}
\label{def}
\end{Def}
The following proposition shows that the definition of $\whf$ makes sense. 
\begin{Prop}
For any $f \in \C_{\A}^{E}$ and any basis $\B$ of $\C^{\A}$, $\whf$ exists and is unique.
\end{Prop}
\begin{proof}
Theorem \ref{Thm:basis} asserts that $\{\E^{[b]} | b \in \B^{E}\}$ is a basis
of $\C_{\A}^{E}$. Hence, $f$ can be uniquely written as $f = \sum_{b \in \B^{E}} c_{b} \E^{[b]}$ where $c_{b} \in \C$ for all $b \in \B^{E}$.
Let $\whf(b) = c_{b}$ then $\whf$ is a map from $\B^{E}$ into $\C$ and it is unique. Therefore, we are done if we show this $\whf$ is consistent
with the definition. But $f = \sum_{b \in \B^{E}} \whf(b) \E^{[b]}$ implies
\begin{eqnarray*}
f(a) \hspace{-.5cm}&&= \sum_{b \in \B^{E}} \whf(b) \E^{[b]}(a) = \sum_{b \in \B^{E}} \whf(b) \E(a,b) 
 = \sum_{b \in \B^{E}} \whf(b) \E_{[a]}(b) 
 = \langle \whf, \E_{[a]} \rangle 
\end{eqnarray*}
as desired.
\end{proof}

%
Let $E_{1}$ and $E_{2}$ be two disjoint finite sets and assume $f_{1} \in \C_{\A}^{E_{1}}$ and $f_{2} \in \C_{\A}^{E_{2}}$. Let $E = E_{1} \cup E_{2}$, then for any $b_{E} \in \B^{E}$, we have
\begin{eqnarray*}
\langle f_{1} \ot f_{2}, \E^{[b_{E}]} \rangle \hspace{-.5cm} &&
= \sum_{a_{E} \in \A^{E}} (f_{1} \ot f_{2})(a_{E}) \E(a_{E}, b_{E}) \\
&&\hspace{0cm} = \sum_{a_{E} \in \A^{E}} (f_{1} \ot f_{2})(a_{E})    (\Ot_{e\in E} \tau_{e}(b_{E}(e))) (a_{E}) \\
&&\hspace{-0cm} = \sum_{a_{E} \in \A^{E}} \!\!\!\!\! f_{1}(a_{E:E_1})  f_{2}(a_{E:E_2})
  (\! \Ot_{e\in E_{1}} \!\! \tau_{e}(b_{E}(e)) \ot \Ot_{e\in E2} \!\! \tau_{e}(b_{E}(e) )  ) (a_{E}) \\
&&\hspace{-0cm} = \sum_{a_{E} \in \A^{E}}  f_{1}(a_{E:E_1})  f_{2}(a_{E:E_2}) 
   (\Ot_{e\in E_{1}}  \tau_{e}(b_{E}(e) ) ) (a_{E:E_1})  \cdot  (\Ot_{e\in E_2} \tau_{e}(b_{E}(e) ) ) (a_{E:E_2})  \\
&&\hspace{-0cm} \stackrel{(a)}{=} \sum_{a_{E:E_{1}} \in \A^{E_{1}}} f_{1}(a_{E:E_{1}})   (\Ot_{e\in E_{1}} \tau_{e}(b_{E}(e)) ) (a_{E:E_{1}}) 
 \sum_{a_{E:E_{2}} \in \A^{E_{2}}} f_{2}(a_{E:E_{2}})    (\Ot_{e\in E_{2}} \tau_{e}(b_{E}(e)) ) (a_{E:E_{2}}) \\
&&\hspace{-0cm} \stackrel{(b)}{=} \sum_{a_{E:E_{1}} \in \A^{E_{1}}} f_{1}(a_{E:E_{1}})   (\Ot_{e\in E_{1}} \tau_{e}(b_{E:E_1}(e)) ) (a_{E:E_{1}}) 
\sum_{a_{E:E_{2}} \in \A^{E_{2}}} f_{2}(a_{E:E_{2}})    (\Ot_{e\in E_{2}} \tau_{e}(b_{E:E_2}(e)) ) (a_{E:E_{2}}) \\
&&\hspace{-0cm}= \langle f_{1}, \E_{1}^{[b_{E:E_{1}}]} \rangle \cdot \langle f_{2}, \E_{2}^{[b_{E:E_{2}}]} \rangle
\end{eqnarray*}
where (a) is due to lemma \ref{Lemma:0} and (b) follows from the fact that $b_{E}(e) = b_{E:E_{i}}(e)$ for all $e \in E_{i}$, $i = 1, 2$. Now induction can be used to show that this holds for any disjoint sets $E_{1}, \ldots, E_{p}$. Hence, we have the following proposition (the second part of the proposition can be shown using a similar argument).
\begin{Prop}
Let $\{ E_{1}, \ldots, E_{p} \}$ be an arbitrary partition of a finite set $E$ and $\B$ be an arbitrary basis of $\C^{\A}$. 
 Further, let $f_{i} \in \C_{\A}^{E_{i}}$ for all $i$. Then for any $b_{E} \in \B^{E}$
\[
\langle \Ot_{1 \leq i \leq p} f_{i}, \E^{[b_{E}]} \rangle = \prod_{1 \leq i \leq p} \langle f_{i}, \E_{i}^{[b_{E:E_{i}}]} \rangle
\]
and for any $a_{E} \in A_{E}$
\[
\langle \Ot_{1 \leq i \leq p} \whf_{i}, \E_{[a_{E}]} \rangle = \prod_{1 \leq i \leq p} \langle \whf_{i}, \E_{i[a_{E:E_{i}}]} \rangle
\]
\label{prop:inner}
\end{Prop}

%
%

%% file: holant.tex
\section{Holant theorem}

\label{sec:holant}

Let $E_{1}, \ldots, E_{p}$ and $E'_{1}, \ldots, E'_{r}$ be two arbitrary partitions of a finite set $E$.
Further let $g_{i} \in \C_{\A}^{E_{i}}$ for all $1 \leq i \leq p$ and let $h_{i} \in \C_{\A}^{E'_{i}}$ for all $1 \leq i \leq r$. 
The following lemma will give a proof of the Holant theorem.
\begin{Lemma}
Let $g_{1}, \ldots, g_{p}$ and $h_{1}, \ldots, h_{r}$ be as above. Then under any basis $\B$ of $\C^{\A}$, the following holds 
\[
\langle \Ot_{1 \leq i \leq p} g_{i}, \Ot_{1 \leq i \leq r} h_{i} \rangle = \langle \Ot_{1 \leq i \leq p} \whg_{i}, \Ot_{1 \leq i \leq r} \wch_{i} \rangle
\]
\label{Lemma:holant}
\end{Lemma}
(note that the bilinear map on the left of the equality is over the domain $\C_{\A}^{E} \times \C_{\A}^{E}$ and the one on the right is over $\C_{\B}^{E} \times \C_{\B}^{E}$).
\begin{proof}
The lemma follows via the following series of equalities:
\begin{eqnarray*}
\langle\Ot_{1 \leq i \leq p} g_{i} ,  \Ot_{1 \leq i \leq r} h_{i}\rangle
\!\!\!\! && \stackrel{(\ref{eq:inner})}{=} \sum_{x_{E} \in \A^{E} } ( \Ot_{1 \leq i \leq r} h_{i} )(x_{E})   ( \Ot_{1 \leq i \leq p} g_{i} )(x_{E}) \\
&& \stackrel{(a)}{=} \sum_{x_{E} \in \A^{E} } ( \Ot_{1 \leq i \leq r} h_{i} )(x_{E})  
 \prod_{1 \leq i \leq p} g_{i}(x_{E:E_{i}}) \\
&& \stackrel{(b)}{=} \sum_{x_{E} \in \A^{E} } ( \Ot_{1 \leq i \leq r} h_{i} )(x_{E}) 
 \prod_{1 \leq i \leq p} \langle\whg_{i}, \E_{i[x_{E:E_{i}}]}   \rangle \\
&& \stackrel{(c)}{=} \sum_{x_{E} \in \A^{E} } ( \Ot_{1 \leq i \leq r} h_{i} )(x_{E}) \ 
  \langle\Ot_{1 \leq i \leq p} \whg_{i},  \E_{[x_{E}]}  \rangle \\
&& \stackrel{(\ref{eq:inner})}{=} \sum_{x_{E} \in \A^{E} } ( \Ot_{1 \leq i \leq r} h_{i} )(x_{E}) 
\!\!\!\! \sum_{y_{E} \in \B^{E}} (\Ot_{1 \leq i \leq p} \whg_{i})(y_{E})  \E_{[x_{E}]}(y_{E})  \\
%
&& \stackrel{(d)}{=} \sum_{x_{E} \in \A^{E} } ( \Ot_{1 \leq i \leq r} h_{i} )(x_{E}) 
 \!\!\!\!\sum_{y_{E} \in \B^{E}} (\Ot_{1 \leq i \leq p} \whg_{i})(y_{E})  \E(x_{E}, y_{E})  \\
&& \stackrel{(d)}{=} \sum_{y_{E} \in \B^{E}} (\Ot_{1 \leq i \leq p} \whg_{i})(y_{E})
\!\!\!\! \sum_{x_{E} \in \A^{E} } ( \Ot_{1 \leq i \leq r} h_{i} )(x_{E}) 
 \E^{[y_{E}]}(x_{E})  \\
&& \stackrel{(\ref{eq:inner})}{=} \sum_{y_{E} \in \B^{E}} (\Ot_{1 \leq i \leq p} \whg_{i})(y_{E})
 \langle \Ot_{1 \leq i \leq r} h_{i},
 \E^{[y_{E}]} \rangle \\
&& \stackrel{(c)}{=} \sum_{y_{E} \in \B^{E}} (\Ot_{1 \leq i \leq p} \whg_{i})(y_{E})
 \prod_{1 \leq i \leq r} \langle h_{i},
 \E_{i}^{[y_{E:E_{i}}]} \rangle \\ 
&& \stackrel{(b)}{=} \sum_{y_{E} \in \B^{E}} (\Ot_{1 \leq i \leq p} \whg_{i})(y_{E})
 \prod_{1 \leq i \leq r} \wch_{i}(y_{E:E_{i}}) \\ 
&& \stackrel{(a)}{=} \sum_{y_{E} \in \B^{E}} (\Ot_{1 \leq i \leq p} \whg_{i})(y_{E})
 (\Ot_{1 \leq i \leq r} \wch_{i}) (y_{E}) \\  
&& \stackrel{(\ref{eq:inner})}{=} \langle\Ot_{1 \leq i \leq p} \whg_{i}, \Ot_{1 \leq i \leq r} \wch_{i}\rangle 
\end{eqnarray*}
where (a), (b) and (c) are due to definition of c-tensor, definition \ref{def} and proposition 
 \ref{prop:inner}, respectively. Definitions of $\E, \E_{[x_{E}]}$ and $\E^{[y_{E}]}$
give equalities labeled by (d).
\end{proof}

Now we state and prove Holant theorem

\begin{Theorem}
Let $g_{1}, \ldots, g_{p}$ and $h_{1}, \ldots, h_{r}$ be as defined earlier. Then for any basis $\B$ of $\C^{\A}$, the following holds
\begin{eqnarray*}
\sum_{x_{E} \in \A^{E}} 
\prod_{i=1}^{p} g_{i}(x_{E:E_{i}}) 
\prod_{i=1}^{r} h_{i}(x_{E:E'_{i}}) 
= 
\sum_{y_{E} \in \B^{E}} \prod_{i=1}^{p} \whg_{i}(y_{E:E_{i}}) 
\prod_{i=1}^{r} \wch_{i}(y_{E:E'_{i}})
\end{eqnarray*}
\end{Theorem}


\begin{proof}
\begin{eqnarray*}
\sum_{x_{E} \in \A^{E}} \prod_{i=1}^{p} g_{i}(x_{E:E_{i}})  \prod_{i=1}^{r} h_{i}(x_{E:E'_{i}})
&\stackrel{(a)}{=}&
\sum_{x_{E} \in \A^{E}}   (\Ot_{1 \leq i \leq p} g_{i} )(x_{E}) (\Ot_{1 \leq i \leq r} h_{i} )(x_{E}) \\
&\stackrel{(\ref{eq:inner})}{=}&
\langle\Ot_{1 \leq i \leq p} g_{i}, \Ot_{1 \leq i \leq r} h_{i}\rangle \\
%
&\stackrel{(b)}{=}&
 \langle\Ot_{1 \leq i \leq p} \whg_{i}, \Ot_{1 \leq i \leq r} \wch_{i}\rangle \\
&\stackrel{(\ref{eq:inner})}{=}&
\sum_{y_{E} \in \B^{E}}   (\Ot_{1 \leq i \leq p} \whg_{i} )(y_{E}) (\Ot_{1 \leq i \leq r} \wch_{i} )(y_{E}) \\
&\stackrel{(a)}{=}&
\sum_{y_{E} \in \B^{E}} \prod_{i=1}^{p} \whg_{i}(y_{E:E_{i}})   \prod_{i=1}^{r} \wch_{i}(y_{E:E'_{i}}) 
\end{eqnarray*}
where (a) is due to definition of c-tensor and (b) is due to Lemma \ref{Lemma:holant}.
\end{proof}

%% file: conclusion.tex
\section{Conclusion}

\label{conclusion}

In this paper, we provide an 
alternative proof of Holant theorem using a commutative notion of tensor 
product. The proof appears to be simpler and more transparent. As 
holographic algorithms are expected to demonstrate their power in many 
problems of information theory, we hope that this work makes this new 
family of algorithms and the notion of holographic reduction more 
accessible to audience in information theory community.

%% file: acknowledgment.tex
\section*{Acknowledgment}
The first author wishes to thank Frank Kschischang for introducing him to the subject of holographic algorithms.